\documentclass[a4paper,fleqn]{cas-sc}

\usepackage[numbers]{natbib}

\def\tsc#1{\csdef{#1}{\textsc{\lowercase{#1}}\xspace}}
\tsc{WGM}
\tsc{QE}

\usepackage{amsmath,amssymb,amsfonts}
\usepackage{graphicx}
\usepackage{booktabs}
\usepackage{multirow}
\usepackage{tabularx}
\usepackage{algorithm}
\usepackage[noend]{algpseudocode}
\usepackage{xspace}
\usepackage{microtype}
\usepackage{float}
\usepackage{placeins}
\usepackage{tablefootnote}
\usepackage{capt-of}

\newdefinition{definition}{Definition}

\newcommand{\bigO}{\mathcal{O}}

\newtheorem{theorem}{Theorem}
\newtheorem{lemma}[theorem]{Lemma}
\newdefinition{remark}{Remark}
\newproof{proof}{Proof}
\newproof{pot}{Proof of Theorem \ref{thm}}
\newtheorem{observation}[theorem]{Observation}
\begin{document}
\let\WriteBookmarks\relax
\def\floatpagepagefraction{1}
\def\textpagefraction{.001}

\setlength{\textfloatsep}{5pt plus 2pt minus 2pt}
\shorttitle{Beer Path Problems in Temporal Graphs}    

\shortauthors{D'Ascenzo, Italiano, Kanellopoulos, Mpanti, Pagourtzis, Pergaminelis}  

\title [mode = title]{Beer Path Problems in Temporal Graphs}  



%
\author[1]{Andrea D'Ascenzo}[orcid=0000-0001-5612-0798]
\ead{andrea.dascenzo@gssi.it}

\author[2]{Giuseppe F. Italiano}[orcid=0000-0002-9492-9894]
\ead{gitaliano@luiss.it}

\author[3,4]{Sotiris Kanellopoulos}[orcid=0009-0006-2999-0580]
\cormark[1]
\ead{s.kanellopoulos@athenarc.gr}

\author[2]{Anna Mpanti}[orcid=0009-0006-8659-7830]
\ead{ampanti@luiss.it}

\author[3,4]{Aris Pagourtzis}[orcid=0000-0002-6220-3722]
\ead{pagour@cs.ntua.gr}

\author[3,4]{Christos Pergaminelis}[orcid=0009-0009-8981-3676]
\cormark[1]
\ead{chrispergaminelis@mail.ntua.gr}

\cortext[1]{Corresponding author}

\affiliation[1]{
  organization={Gran Sasso Science Institute},
  city={L'Aquila},
  country={Italy}
}

\affiliation[2]{
  organization={LUISS University},
  city={Rome},
  country={Italy}
}

\affiliation[3]{
  organization={National Technical University of Athens},
  city={Athens},
  country={Greece}
}

\affiliation[4]{
  organization={Archimedes, Athena Research Center},
  city={Athens},
  country={Greece}
}

















\nonumnote{}

\begin{abstract}
Computing paths in graph structures is a fundamental operation in a wide range of applications, from transportation networks to data analysis. The \emph{beer path} problem, which captures the option of visiting points of interest, such as gas stations or convenience stores, prior to reaching the final destination, has been recently introduced and extensively studied in static graphs. However, existing approaches do not account for temporal information (for example,  transit service schedules), which is often crucial in real-world scenarios. 
In this work, we introduce the notion of beer paths in temporal graphs, where edges are time-dependent and \emph{beer vertices} are active only at specific times. We formally define the problems of computing earliest-arrival, latest-departure, fastest, and shortest temporal beer paths and propose efficient algorithms for all of them under both edge stream and adjacency list representations. We show that the time complexity of each of our algorithms is aligned with that of the corresponding standard temporal pathfinding algorithm, thus preserving efficiency. 
We also present preprocessing techniques that enable efficient query answering under dynamic conditions, for example new openings or closings of shops. We achieve this either through appropriate precomputation of selected paths or by transforming a temporal graph into an equivalent static graph.
\end{abstract}





\maketitle

\section{Introduction}
Answering path queries is a fundamental operation on networks, with numerous optimization and information retrieval tasks relying on an efficient response to such queries. 
Bacic, Mehrabi and Smid~\cite{Bacic_beer_path} recently introduced the concepts of \emph{beer vertices} and \emph{beer paths} to model the option of taking a detour towards a set of \emph{points of interest} when traveling from a source point to a destination point.
In fact, choosing detours is an essential part of daily life and modern computing systems. When traveling from point~A to point~B, we often need to adjust our route, perhaps to refuel at a gas station or to pick up a beer.
This can be modeled as follows: Given a weighted graph $G=(V,E)$ and a subset $B \subseteq V$ of beer vertices, a beer path
between two vertices $u$ and $v$ is a path that starts from $u$, ends at $v$, and visits at least one beer vertex
in $B$. The \emph{beer distance} between two vertices is the shortest length of any beer path connecting them.

Since their introduction, beer paths have been extensively studied in different graph classes and frameworks (e.g.,~\cite{Bacic_beer_path,Das_beer_path_interval,Hanaka_beer_path_decomposition,Gudmundsson_beer_path_treewidth}).
All existing techniques, however, are applicable to \textit{static} graphs. Because of the inherent nature of beer path applications, temporal information is, in fact, a fundamental feature to consider. For example, bus routes may operate only at specific times, while beer shops may be subject to opening and closing hours. Temporal graphs extend the concept of static graphs to include such information, with each edge being associated with a starting time and a traversal time, while two vertices are adjacent only in the time instants in which an edge connecting them is active. 
In this work we introduce and study beer path problems in temporal graphs.

\subsection{Related work}
Temporal graphs have been the focus of research in several fields due to their ability to model time-dependent information in a wide range of real-world scenarios. The concepts of earliest-arrival, latest-departure, fastest, and shortest temporal paths were introduced by Wu et al. in~\cite{Path_Wu} and further studied in~\cite{Wu_efficient_paths}. 
Beyond these standard temporal path objectives, Bentert et al.~\cite{Himmel_optimal_walk} study optimal temporal walks under waiting-time constraints, while Erlebach, Hoffmann, and Kammer~\cite{Erlebach_exploration} consider temporal walks in the context of graph exploration (cf.~\cite{DBLP:conf/icalp/DogeasEKMM24}). Temporal walks and paths have also been used in centrality and betweenness analysis~\cite{oettershagen2022temporal,Naima2025temporal,zhang2024efficient}.
Moreover, temporal graphs have recently been studied from perspectives such as pathfinding and network dynamics~\cite{kostakos_temporal_graphs,Michail_temporal_graphs,Spirakis_dynamic}, temporal reachability~\cite{DBLP:conf/isaac/Doring25, Deligkas_reachability_ijcai, Deligkas_temporal_parameterized, Deligkas_shifting} and temporal edge and path covers~\cite{Deligkas_temporal_edge_cover, temporal_path_cover_Dailly_MFCS, cioni2026temporalpathcoversdilworth}. Other recent research includes temporal flows and cuts~\cite{Akrida_temporal_flow}, temporal network optimization under connectivity constraints~\cite{Mertzios_temp_conn} and surveys of temporal network models, algorithms, tools, and applications in areas such as biology and medicine~\cite{hosseinzadeh2022temporal}. Comprehensive overviews on temporal graphs are provided in~\cite{Holme_Temporal_graphs,Wang_Temporal_graphs}.



%
The problem of computing \textit{shortest beer paths} in static graphs was recently introduced by Bacic et al.~\cite{Bacic_beer_path}.
Since the problem is generally tractable, researchers have focused on specific graph classes to design data structures that allow answering beer path queries faster than the baseline method. Bacic et al.~\cite{Bacic_beer_path} developed an index for \textit{outerplanar graphs}; Das et al.~\cite{Das_beer_path_interval} studied the beer path problem on \textit{interval graphs}; Hanaka et al.~\cite{Hanaka_beer_path_decomposition} achieved optimal query time on \textit{series-parallel graphs} and linear preprocessing time on graphs having \textit{bounded-size triconnected components};  Gudmundsson and Sha~\cite{Gudmundsson_beer_path_treewidth} focused on bounded treewidth graphs; 
Coudert, D'Ascenzo and D'Emidio~\cite{Coudert_beer_path_indexing} designed an index data structure for answering beer path queries efficiently.
Additionally, Bilò et al.~\cite{bilo_spanner_beer} address the problem of
constructing graph spanners w.r.t. the \textit{group Steiner metric}, which generalizes the beer distance metric. Hon, Huang and Li~\cite{Hon_beer_flow} recently introduced the maximum beer flow problem, extending beer paths to a flow setting in which every unit of flow must visit a beer store.
Note that the shortest beer path problem is a special case of the \textit{Generalized Shortest Paths} problem~\cite{Rice_Generalized_Shortest} and the NP-hard \textit{Generalized Travelling Salesperson Problem} 
(cf.~\cite{laporte_GTSP,Pop_GTSP,Rice_GTSP,srivastava_GTSP}). Traveling salesman variants have been studied in temporal graphs by Michail and Spirakis~\cite{Michael_temp_tsp}. To the best of our knowledge, this is the first study of beer path problems in temporal graphs.

\subsection{Our contribution}
In this paper, we introduce the concept of beer paths in (directed) temporal graphs and several related optimization problems, namely \emph{earliest-arrival}, \emph{latest-departure}, \emph{fastest}, and \emph{shortest temporal beer paths} (EABP, LDBP, FBP, SBP). 
In order to further incorporate temporal information that accurately reflects real-life applications, our framework also considers timestamped beer vertices, i.e., the case where each beer vertex is augmented with the time instants in which it is active as a point of interest. 
This feature naturally models many scenarios, e.g., opening and closing times of bars or gas stations, active motorway toll booths, instants in which airports may serve for emergency landing, and many more. 
Note that our algorithms can also work for active time intervals instead of instants with minor modifications (by using binary searches inside time intervals instead of lists of time instants). Most of our algorithms receive a sorted \emph{edge stream} as input, as is standard for temporal graphs (cf.~\cite{Path_Wu,Holme_Temporal_graphs,Lee_edgestream,OettershagenKM23_edgestream}). However, we also consider algorithms with random access to the edges of the graph, for example through an adjacency list with temporal logs (cf. Caro, Rodríguez and Brisaboa~\cite{Caro_temporal_logs}), when they can offer a better time complexity.

\begin{center}
\begin{minipage}{\textwidth}
\centering
    \begin{tabular}{llll}
\toprule
    Objective & Source/Target & Input Structure & Time Complexity \\ \hline
    \multirow{2}{*}{EABP} & One-to-all & Edge Stream & $\bigO(n+M)$ \\ 
     & One-to-all & Adjacency List & $\bigO(m\log \pi+n\log n)$ \\ \hline

    \multirow{2}{*}{LDBP} & All-to-one & Edge Stream (reverse) & $\bigO(n+M)$ \\
     & All-to-one & Adjacency List & $\bigO(m\log \pi+n\log n)$ \\ \hline
 
    FBP & One-to-all & Edge Stream & $\bigO(n+(M+kc)\log d_{in})$ \\ \hline

     \multirow{2}{*}{SBP} & \multirow{2}{*}{One-to-one} & Edge Stream & \multirow{2}{*}{$\bigO(n+M \log d_{max} + k d_{in} \log d_{out})$} \\
     & & (regular and reverse) & \\ \hline
 
    \end{tabular}
    \captionof{table}{Overview of our algorithms in Section~\ref{sec:algo}. Some complexity terms that are negligible under reasonable assumptions are omitted from this table for simplicity. See Table~\ref{table:notation} for a comprehensive list of all variable symbols used throughout the paper.}
    \label{tab:results_1}
\end{minipage}
\end{center}

In Section~\ref{sec:algo}, we propose efficient algorithms for the aforementioned objectives, with each one differing both in methods and in complexity (see Table~\ref{tab:results_1} for a synopsis). As intermediate steps for computing EABPs and LDBPs, we define \emph{multiple-source} earliest-arrival paths and \emph{multiple-target} latest-departure paths and propose algorithms for computing them, which can be of independent interest. 
For FBP and SBP we rely on the computation of certain non-dominated paths via various domination criteria. For SBP in particular, we define a new path domination criterion and propose an algorithm for computing respective non-dominated paths. The time complexity of our algorithms generally aligns with that of respective standard temporal pathfinding algorithms~\cite{Path_Wu}. All our algorithms except the one for FBP preserve the original graph structure, which is beneficial because, depending on the data structure used, adding new edges to a temporal graph may be a costly operation.


In Section~\ref{sec:preprocessing}, we propose preprocessing methods for answering beer path queries faster. For EABP and LDBP, this is achieved through the precomputation of selected non-dominated paths in the graph. For FBP and SBP, we utilize the transformation proposed by Wu et al.~\cite{Path_Wu} to convert a temporal graph to a static graph. 
We remark that the query time achieved for SBP is faster than the respective algorithm of Section~\ref{sec:algo} only for graphs with strictly positive traversal times. 
See Table~\ref{tab:results_2} for a synopsis of these results.


\begin{center}
\begin{minipage}{\textwidth}
\centering
    \begin{tabular}{llllll}
    \hline
    Objective & Method & Structure Time & Space & Source/Target  & Query Time 
    \\ \hline
    EABP & Path precomputation & $\bigO(n^2 + nM \log c)$ & $\bigO (nkc)$ & One-to-one  & $\bigO (k \log c)$ \\ \hline
     
    LDBP & Path precomputation & $\bigO(n^2 + nM \log c)$ & $\bigO (nkc)$ & One-to-one  & $\bigO (k \log c)$ \\ \hline
     
    FBP & Conversion to static & $\bigO(M\log d_{max})$ & $\bigO(M)$ & One-to-all  & $\bigO(M)$ \\ \hline
    
    \multirow{2}{*}{SBP} & \multirow{2}{*}{Conversion to static} & \multirow{2}{*}{$\bigO(M\log d_{max})$} & \multirow{2}{*}{$\bigO(M)$} & \multirow{2}{*}{One-to-all}  & $\bigO(M\log M)$ \\ 

    & & & & & $\bigO(M)$ \footnotemark \\ \hline
    
    \end{tabular}
    \captionof{table}{Overview of our preprocessing methods in Section~\ref{sec:preprocessing}.}
    \label{tab:results_2}
\end{minipage}
\end{center}

\footnotetext{\(\mathcal{O}(M)\) time can be achieved if all temporal edges have strictly positive traversal times (\(\lambda>0\)), in which case the transformed graph is a DAG. See Section~\ref{sec:SBP_transformed} for more details.} 

\section{Preliminaries and problem definition}\label{sec:prelims}

\subsection{Preliminaries}

Let $ G = (V, E)$  be a directed temporal graph, with $|V|=n$ and $|E|=M$. Each temporal edge $ e \in E $ is defined as a quadruple $ (u, v, t, \lambda )$, with \( u, v \in V \), where  $t$  denotes the \textit{starting time}, $ \lambda$  denotes the \textit{traversal time} needed to travel from $ u$ to $v$ starting at time $t$, and $t + \lambda $ represents the \textit{ending time}. Most of our algorithms receive as input an \emph{edge stream}, containing all edges sorted by ascending $t$-value. This is standard for temporal graphs (cf.~\cite{Holme_Temporal_graphs,Lee_edgestream,OettershagenKM23_edgestream,Path_Wu}).

The set of temporal edges from $u$ to $v$ $(u,v \in V)$ is denoted by $E(u,v)$. Let $\pi = \max \limits _{u,v \in V}\{|E(u,v)|\}$ (maximum number of parallel edges) and $m = |\{u,v \in V : |E(u,v)| > 0\}|$ (number of connected pairs of vertices).
We denote with $N_{out}(u) = \{v : (u,v,t,\lambda) \in E \}$ the set of out-neighbours of $u \in V$, and with $d_{out}$ the maximum out-degree in $G$. We similarly define $N_{in}(u)$ and $d_{in}$.
Finally, $d_{max}$ denotes the maximum between $d_{out}$ and $d_{in}$.

A \emph{temporal path} $P$ is defined as a sequence of edges $e_i=(v_i, v_{i+1},t_i, \lambda_i)\in E$, with $(t_i+\lambda_i) \leq t_{i+1}$, for $1 \leq i \leq \ell$, where $\ell$ is the number of edges in the path. Throughout the paper we may informally denote a path by the sequence of vertices it traverses (instead of its edges), with the notation $P=\langle v_1,v_2,\dots, v_{\ell+1}\rangle$.
The quantity $(t_\ell + \lambda_\ell)$ is known as the \emph{ending time} of $P$, denoted by $\text{end}(P)$, while $t_1$ is the \emph{starting time} of $P$, $\text{start}(P)$. 
The \emph{duration} of $P$ is defined as $\text{dura}(P)=\text{end}(P)-\text{start}(P)$. 
Finally, the sum over the traversal times of the edges of $P$ is known as the \emph{distance} of the path, denoted as $\text{dist}(P)$. For a comprehensive list of all notation used throughout the paper, see Table~\ref{table:notation}.

\begin{table}[H]
    \centering
    \begin{tabularx}{\textwidth}{|c|X|}
    \hline
    \textbf{Notation} & \textbf{Description} \\
    \hline
    
    $G = (V, E)$ & A (directed) temporal graph with vertex set $V$ and temporal edge set $E$. \\
    \hline

    $B$ & The beer vertex set in beer path problems. \\
    \hline
    
    $e = (u, v, t, \lambda)$ & A temporal edge from vertex $u$ to vertex $v$, with starting time $t$ and traversal time $\lambda$. \\
    \hline

    $n$ & Number of vertices in a temporal graph. \\
    \hline

    $k$ & Number of beer vertices in a temporal graph. \\
    \hline

    $M$ & Number of temporal edges in a temporal graph. \\
    \hline

    $m$ & Number of vertex pairs $u,v$ such that at least one temporal edge exists from $u$ to $v$. \\
    \hline

    $\pi$ & Maximum number of parallel temporal edges (over all vertex pairs). \\
    \hline

    $T_b$ & A list of timestamps in which the beer vertex $b\in B$ is active. \\
    \hline

    $T$ & The maximum size of lists $T_b$ over all $b\in B$. \\
    \hline

    $N_{out}(v)$ & The set of out-neighbours of vertex $v$. \\
    \hline

    $N_{in}(v)$ & The set of in-neighbours of vertex $v$. \\
    \hline

    $d_{out}$ & The maximum out-degree over all vertices. \\
    \hline

    $d_{in}$ & The maximum in-degree over all vertices. \\
    \hline

    $d_{max}$ & $\max \{d_{out},d_{in}\}$. \\
    \hline

    $c$ & $\min \{d_{out},d_{in}\}$. \\
    \hline

    $\text{start}(P)$ & Starting time of temporal path $P$ from its first vertex. \\
    \hline
    
    $\text{end}(P)$ & Arrival time of temporal path $P$ to its last vertex. \\
    \hline
    
    $\text{dura}(P)$ & Duration of temporal path $P$, i.e., $\text{end}(P)-\text{start}(P)$.\\
    \hline
    
    $\text{dist}(P)$ & Distance of temporal path $P$, i.e., the sum over the $\lambda$-values of the edges in $P$. \\
    \hline

    \multirow{2}{*}{$[t_\alpha, t_\omega]$} & The time window desired for solutions in some temporal path problem. For path $P$ to be considered, it must be $\text{start}(P) \geq t_\alpha$ and $\text{end}(P) \leq t_\omega$. \\
    \hline

   $P_B(x,y,[t_\alpha,t_\omega])$ & The set of temporal beer paths from vertex $x$ to vertex $y$ within time interval $[t_\alpha,t_\omega]$. \\
    \hline

    EABP & Earliest-arrival (temporal) beer path ($\text{end}(P)$ is minimized). \\
    \hline

    LDBP & Latest-departure (temporal) beer path ($\text{start}(P)$ is maximized). \\
    \hline

    FBP & Fastest (temporal) beer path ($\text{dura}(P)$ is minimized). \\
    \hline

    SBP & Shortest (temporal) beer path ($\text{dist}(P)$ is minimized). \\
    \hline

    MSEAP & Multiple-source earliest-arrival (temporal) path (Def.~\ref{def:multisource}). \\
    \hline

    MTLDP & Multiple-target latest-departure (temporal) path (Def.~\ref{def:multitarget}). \\
    \hline

    \end{tabularx}
    \caption{Frequently used notation.}
    \label{table:notation}
\end{table}

Throughout the paper, we may assume the absence of \textit{dominated} edges for some of our algorithms.


\begin{definition}[\textbf{Dominated edge}]
\label{def:dom_edge}
    We say that a temporal edge $e = (u, v, t, \lambda)$ is \emph{dominated} if there exists another edge $e' = (u, v, t', \lambda')$ such that $t' \ge t$ and $t' + \lambda' \le t + \lambda$, with at least one of the two inequalities being strict. 
\end{definition}


Intuitively, for all optimality criteria considered in this paper there would be no benefit in using a dominated temporal edge $e$ as opposed to a temporal edge $e'$ that dominates $e$. Note that this does not extend to all temporal path objectives in the literature; for example, an optimal path under the minimum waiting time objective~\cite{Bentert_waiting_time} may use dominated edges.
The edge domination criterion naturally generalizes to paths.

\begin{definition}[\textbf{Dominated paths}]
\label{def:dom_path}
Let $P$ be a temporal $u$--$v$ path. We say that $P$ is \emph{dominated} if there exists another temporal $u$--$v$ path $P'$ such that $\text{start}(P') \ge \text{start}(P)$ and $\text{end}(P') \le \text{end}(P)$, with at least one of the inequalities being strict. We say that $P$ is \emph{distance-wise dominated} if there exists a temporal $u$--$v$ path $P'$ such that $\text{dist}(P') \le \text{dist}(P)$ and $\text{end}(P') \le \text{end}(P)$, with at least one of the inequalities being strict.
\end{definition}

In some of our algorithms we will use subroutines that compute all non-dominated paths (excluding ties) from one node to all others within a given time interval. This is useful for finding fastest or shortest paths (depending on which domination criterion is used), as well as for precomputing certain temporal paths to answer queries faster. These subroutines (Algorithms~\ref{alg:non_dominated_paths} and~\ref{alg:distance_non_dominated_paths}) stem from the fastest and shortest path algorithms of Wu et al.~\cite{Path_Wu}. We refer the reader to Section~\ref{auxiliary_section_dom_path_algo} for a detailed description of these subroutines and their running times.


\subsection{The temporal beer path problem}

We assume a (directed) temporal graph with a set $B = \{b_1, \dots, b_k\} \subseteq V$ of $k$ \emph{beer vertices}. 
For each beer vertex $b \in B$, we are given a set $T_{b}$ of timestamps denoting the time instants in which $b$ is \emph{active}.
We assume that beer vertices can still be traversed in time units in which they are inactive, functioning as regular vertices.
We define $T=\max\limits_{b \in B} \{|T_b|\}$.

Given a temporal graph $G$, vertices $x,y \in  V$ and a time interval $[t_\alpha, t_\omega ]$, we denote the set of temporal beer paths between two vertices $x$ and $y$ starting at or after $t_\alpha$ and ending no later than~$t_\omega$ as $P_B(x,y,[t_\alpha, t_\omega]) =$ $\{P=\langle x=v_1,v_2,\dots,v_\ell=y\rangle:P$  is a temporal path from $x$ to $y$ such that  $start(P) \geq t_\alpha, \ end(P) \leq t_\omega, \ \exists \ (i \in [\ell], t \in T_{v_i}) : (v_i \in B, \  t\in [t_{i-1}+\lambda_{i-1},t_i]  ) \}$. Note that a beer path must traverse at least one beer vertex at a time in which it is active (e.g., the $x$--$b_5$--$y$ path in Figure~\ref{fig:temporal_graph} is not a valid temporal beer path).

\begin{definition}[Objectives]\label{def:beer_paths} 
A temporal $x$--$y$ path $P \in P_B(x,y,[t_\alpha,t_\omega])$ is:
\begin{itemize}
    \item An \emph{earliest-arrival beer path} (\textbf{EABP}) if 
    $\text{end}(P) = \min \{ \text{end}(P') : P' \in P_B(x,y,[t_\alpha,t_\omega]) \}$.
    
    \item A \emph{latest-departure beer path} (\textbf{LDBP}) if 
    $\text{start}(P) = \max \{ \text{start}(P') : P' \in P_B(x,y,[t_\alpha,t_\omega]) \}$.
    
    \item A \emph{fastest beer path} (\textbf{FBP}) if 
    $\text{dura}(P) = \min \{ \text{dura}(P') : P' \in P_B(x,y,[t_\alpha,t_\omega]) \}$.
    
    \item A \emph{shortest beer path} (\textbf{SBP}) if 
    $\text{dist}(P) = \min \{ \text{dist}(P') : P' \in P_B(x,y,[t_\alpha,t_\omega]) \}$.
\end{itemize}
\end{definition}

\begin{figure}[pos=t]
\centering
\includegraphics[width=0.55\textwidth]{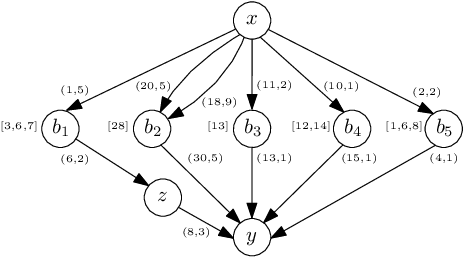}
\caption{Example in which the $x$--$y$ EABP/LDBP/FBP/SBP use $b_1, \ldots, b_4$ respectively. Numbers in parentheses denote $(t,\lambda)$ of edges and numbers in brackets denote beer active times. The $x$--$y$ path that uses $b_5$ is not a valid temporal beer path. The edge with $(t,\lambda)=(18,9)$ is dominated and is irrelevant for all four objectives.
}
\label{fig:temporal_graph}
\end{figure}






For examples for each objective, see Figure~\ref{fig:temporal_graph}. In all our algorithms, we focus on the computation of the objective function (e.g., the earliest arrival time) instead of reporting the corresponding optimal path. This  path can be easily recovered by storing parent information. 

\section{Synopsis of previous results}

\subsection{Non-dominated path algorithms from Wu et al.~\cite{Path_Wu}}
\label{auxiliary_section_dom_path_algo}

In this section we give the full descriptions of Algorithms~\ref{alg:non_dominated_paths} and~\ref{alg:distance_non_dominated_paths} for computing non-dominated or distance-wise non-dominated paths from a given source respectively. These algorithms consist of parts of the fastest and shortest path algorithms of Wu et al.~\cite{Path_Wu}. We refer the reader to that paper for details about their correctness. In this work, we use these algorithms as subroutines for computing FBPs and SBPs, as well as for preprocessing the graph to answer certain queries faster.

Note that in the case of a tie between two $x-v$ temporal paths $P,\ P'$ (e.g. $start(P)=start(P')$ and $end(P)=end(P')$ in Alg.~\ref{alg:non_dominated_paths}), the algorithm only stores one of them, since they can be used interchangeably for any objective. Throughout the paper, when we say "all non-dominated paths", we imply that only one path is kept when there are such ties.

Algorithm~\ref{alg:non_dominated_paths} computes all non-dominated paths from a given source within a given time interval. For each candidate path $P$, the algorithm stores a tuple $(s,a)$, where $s=start(P)$ and $a=end(P)$. Algorithm~\ref{alg:non_dominated_paths} runs in time $\bigO(n+M \log c)$, where $c= \min \{d_{in},d_{out}\}$. Note that $c$ is a bound for the amount of non-dominated pairs $(s,a)$ (start, arrival) between any two vertices and, thus, also a bound for the size of each list~$L_v$. For more details, we refer the reader to~\cite{Path_Wu}.

\begin{algorithm}[ht]
\caption{$\texttt{Non-dom\textunderscore paths($G,x,[t_\alpha,t_\omega]$)}$}
\label{alg:non_dominated_paths}
\begin{algorithmic}[1]
\Require{A temporal graph $G = (V, E)$ in edge stream representation, source vertex $x$, time interval $[t_\alpha, t_\omega]$.}
\Ensure{For each $v \in V$, the list $L_v$ of non-dominated $(s, a)$ pairs (from $x$ to $v$) within \([t_{\alpha},t_{\omega}]\).}
\For{$v \in V$}
   \State Create a sorted list $L_v$, where an element of $L_v$ is a pair $(s, a)$ in which $s$ is the starting time of a path $P$ from $x$ to $v$ and is used as the key for ordering in $L_v$, and $a$ is the time in which the path $P$ arrives at $v$; initially, $L_v$ is empty
\EndFor

\For{$e = (u, v, t, \lambda)$ in the edge stream}
    \If{$t \ge t_\alpha$ and $t + \lambda \le t_\omega$}
        \If{$u = x$ and $(t, t) \not\in L_x$}
            \State Insert $(t, t)$ into $L_x$\;
        \EndIf
        
    \State Let $(s', a')$ be the element in $L_u$ for which $a' = \max\{a \mid (s, a) \in L_u,\ a \le t\}$
    
      \If{$\exists (s,a) \in L_v$ such that $s=s'$ and $a > t + \lambda$}
            \State Replace $(s,a)$ with $(s',t + \lambda)$\;
        \Else
            \State Insert $(s', t + \lambda)$ into $L_v$\;
       \EndIf
    
    \State Remove dominated elements in $L_v$\;
    \EndIf
\EndFor

\State \Return $L_v$ for all $v \in V$
\end{algorithmic}
\end{algorithm}

We next present Algorithm~\ref{alg:distance_non_dominated_paths}, which is analogous to Alg.~\ref{alg:non_dominated_paths} for \emph{distance-wise} non-dominated paths. For a candidate path $P$, this algorithm stores a tuple $(d,a)$, where $d=dist(P)$ and $a=end(P)$. Algorithm~\ref{alg:distance_non_dominated_paths} runs in time $\bigO(n+M \log d_{in})$. Observe that $d_{in}$ is a bound for the amount of distance-wise non-dominated pairs $(d,a)$ (distance, arrival) between any two vertices and, thus, also a bound for the size of each list $L_v$ in Alg.~\ref{alg:distance_non_dominated_paths}. For more details, we refer the reader to~\cite{Path_Wu}.

\begin{algorithm}[ht]
\caption{$\texttt{Dist\textunderscore non-dom\textunderscore paths($G,x,[t_\alpha,t_\omega]$)}$}
\label{alg:distance_non_dominated_paths}
\begin{algorithmic}[1]
\Require{A temporal graph $G = (V, E)$ in edge stream representation, source vertex $x$, time interval $[t_\alpha, t_\omega]$.}
\Ensure{For each $v \in V$, the list $S_v$ of distance-wise non-dominated $(d, a)$ pairs (from $x$ to $v$) within \([t_{\alpha},t_{\omega}]\).}
\For{$v \in V$}
    \State Create a sorted list $L_v$, where an element of $L_v$ is a pair $(d, a)$ in which $d$ is the distance of a path $P$ from $x$ to $v$ and $a$ is the time in which path $P$ arrives at $v$ and is used as the key for ordering in $L_v$;  initially, $L_v$ is empty;\;
\EndFor
\For{$e = (u, v, t, \lambda)$ in the edge stream}
    \If{$t \ge t_\alpha$ and $t + \lambda \le t_\omega$}
        \If{$u = x$ and $(0, t) \not\in L_x$}
            \State Insert $(0, t)$ into $L_x$\;
        \EndIf
    \State Let $(d', a')$ be the element in $L_u$ for which $a' = \max\{a \mid (d, a) \in L_u,\ a \le t\}$
    \If{$\exists (d,a) \in L_v$ such that $a=t + \lambda$ and $d > d' + \lambda$}
            \State Replace $(d,a)$ with $(d' + \lambda,t + \lambda)$\;

        \Else
            \State Insert $(d' + \lambda,t + \lambda)$ into $L_v$\;
        \EndIf

    \State Remove distance-wise dominated elements in $L_v$\;
    \EndIf
\EndFor
\State \Return{$L_v$ for all $v \in V$}
\end{algorithmic}
\end{algorithm}

\subsection{The graph transformation of Wu et al.~\cite{Path_Wu}}
\label{auxiliary_section_graph_transformation}

In this section we describe in detail the graph transformation of Wu et al.~\cite{Path_Wu}, which converts a temporal graph $G = (V, E)$ into an equivalent static weighted graph $\widetilde{G} = (\widetilde{V}, \widetilde{E})$.

\begin{figure}[pos=t]
\centering
\includegraphics[scale=0.89]{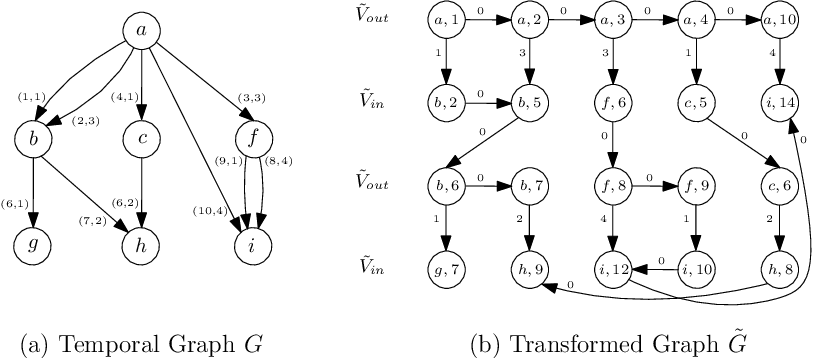}
\caption{Graph transformation, from $G$ in (a) to $\widetilde{G}$ in (b).}
\label{fig:graph_transf}
\end{figure}

For each vertex $v \in V$, the transformation creates two sets of time-stamped nodes: $\widetilde{V}_{in}(v)$, which includes a node $(v, t)$ for every unique time $t$ at which an edge arrives at $v$, and $\widetilde{V}_{out}(v)$, which includes a node $(v, t)$ for every time $t$ at which an edge departs from $v$. These sets are sorted in descending time order. Edges are then created in three categories. First, for each node in $\widetilde{V}_{\text{in}}(u)$, a zero-weight edge is created to the earliest node $(v, t_{out})$ in $\widetilde{V}_{out}(v)$ that has a time greater than or equal to it, if no edge from any other $(u, t'_{\text{in}}) \in \widetilde{V}_{in}(u) $ to $(v, t_{out})$ has been created.
Second, zero-weight edges are added between consecutive time-stamped nodes within $\widetilde{V}_{\text{in}}(v)$ and within $\widetilde{V}_{\text{out}}(v)$ to allow temporal ``waiting'' at the same vertex. Finally, for each original temporal edge $(u, v, t, \lambda) \in E$, a directed edge with weight $\lambda$ is added from $(u, t) \in \widetilde{V}_{\text{out}}(u)$ to $(v, t+\lambda) \in \widetilde{V}_{\text{in}}(v)$. Therefore, the graph $\widetilde{G}$ has a size complexity of $\bigO(M)$ in both the number of vertices and edges. Figures~\ref{fig:graph_transf}(a),(b) show a temporal graph $G$ and its transformed graph $\widetilde{G}$, respectively.

The complexity of the transformation from the temporal graph $G=(V,E)$ to the static weighted graph $\widetilde{G}$ primarily depends on two steps: vertex creation and edge creation. During vertex creation, for each vertex $v \in V$, time-stamped copies are created corresponding to distinct arrival and departure times. Sorting these time instances for all vertices takes $\bigO(M \log (d_{\max}))$ time, where $d_{\max}$ is the maximum degree (either in-degree or out-degree) of any vertex in $G$. For edge creation, zero-weight edges are constructed between consecutive time-stamped nodes and from arrival to departure nodes, which requires $\bigO(M)$ time. Additionally, each original temporal edge induces a weighted edge in $\widetilde{G}$. Placing these edges in $\widetilde{G}$ also requires $\bigO(M \log (d_{max}))$ time. Therefore, the overall time complexity of the graph transformation is $\bigO(M \log (d_{\max}))$.

\section{Algorithms for temporal beer path problems}\label{sec:algo}

\subsection{Algorithms for earliest-arrival beer paths }\label{subsec:earliest_beer}

In this subsection we propose an EABP algorithm (from source $x$ to all nodes), based on the earliest-arrival time algorithm of~\cite{Path_Wu} with a modification for computing \emph{multiple-source} earliest-arrival paths (MSEAP - Def.~\ref{def:multisource}).
We also present an alternative algorithm for graphs containing no dominated edges (a reasonable assumption, since dominated edges are irrelevant for all the objectives considered in this paper).
The latter is similar in spirit to Dijkstra's algorithm and uses binary search for parallel edges, achieving a running time independent of~$M$. Although faster than the former, it requires random access to the edges of the graph through an adjacency list (cf.~\cite{Caro_temporal_logs}), instead of scanning an edge stream in a specific order. We remark that our algorithms do not modify the graph (in contrast to, e.g., using dummy vertices), which is important because inserting new edges in temporal graphs may be a costly operation.


\subsubsection{An earliest-arrival beer path algorithm in edge stream representation }

We define Multiple-Source Earliest-Arrival Paths (MSEAPs), which will prove useful as an intermediate step for computing EABPs. Intuitively, an MSEAP is a path that can reach a target node in the earliest possible time unit, starting from any node of the graph, with restrictions for the starting time of each node. If some $\text{init}[v]$ is $\infty$, we cannot start a path from $v$.

\begin{definition}[MSEAP]\label{def:multisource}
    Given a temporal graph $G=(V,E)$ and initial times $\text{init}[v]\in \mathbb{R}^+ \cup \{\infty\}$, $v\in V$, we say that a temporal path $x$--$y$ starting from $x$ no earlier than $\text{init}[x]$ is a \emph{multiple-source earliest-arrival path} to $y$ if its arrival time at $y$ is minimum among all temporal paths $v$--$y$, $v\in V$, starting from $v$ no earlier than $\text{init}[v]$. We call this arrival time the \emph{multiple-source earliest-arrival path time} of $y$.
\end{definition}

We first show that a generalization of the earliest-arrival path algorithm of~\cite{Path_Wu} can compute MSEAPs. Algorithm~\ref{alg:earliest_multisource_stream} scans each temporal edge once and thus runs in time $\bigO(n+M)$. Note that with initialization $\text{init}[x] = t_\alpha$ for some $x \in V$ and $\text{init}[v] = \infty$ for $v \in V \setminus \{x\}$, Algorithm~\ref{alg:earliest_multisource_stream} coincides with the earliest-arrival time algorithm of~\cite{Path_Wu} with source node $x$ and computes single-source earliest arrival paths.

\begin{algorithm}[ht]
\caption{$\texttt{MSEAP$(G,[t_\alpha,t_\omega],\text{init})$}$}
\label{alg:earliest_multisource_stream}
\begin{algorithmic}[1]

\Require{A temporal graph $G$ in its edge stream
 representation, a time window $[t_{\alpha}, t_{\omega}]$ and initial times $\text{init}[v]$, such that $t_\alpha \leq \text{init}[v] \leq t_\omega$ or $\text{init}[v] = \infty$, for all $v \in V$.}
\Ensure{The MSEAP times within $[t_\alpha, t_\omega]$ to every vertex, assuming we can start at time $\text{init}[v]$ from each $v \in V$.}

\State $\tau[v] \gets \text{init}[v]$ \textbf{for all} $v \in V$

\For{$e=(u,v,t,\lambda)$ in the edge stream}{
    \If{$\tau[u]\leq t$ \textnormal{\textbf{and}} $t+\lambda\leq t_{\omega}$}
       \State $\tau[v] \gets \min\{\tau[v], t+\lambda\}$
    \EndIf
}
\EndFor

\State \Return{$\tau[v]$ for each $v \in V$}
\end{algorithmic}
\end{algorithm}




\begin{lemma}\label{lem:earliest_multisource_prefix_property}
     For $G=(V,E)$ and $\text{init}[v]\in \mathbb{R}^+ \cup \{\infty\}$, $v\in V$, it holds that for every $y \in V$ with finite MSEAP time there exists some MSEAP $\langle v_0,v_1,\ldots,v_p=y \rangle$ to $y$ for which every prefix-subpath $\langle v_0,v_1,\ldots,v_i \rangle$ is an MSEAP to $v_i$.
\end{lemma}

\begin{proof}
    Let \(P=\langle u_0,u_1,\dots ,u_p=y\rangle\) be an arbitrary MSEAP to $y$. If the desired property does not hold for \(P\), take the maximum $i \in [p]$ s.t. the prefix-subpath \(R=\langle u_0,\dots ,u_i\rangle\) of $P$ is not an MSEAP to \(u_i\). Let \(Q=\langle w_0,w_1,\dots ,w_j=u_i\rangle\) be an MSEAP to $u_i$. Replace $R$ with \(Q\), keeping the suffix $\langle u_{i+1},\dots ,u_p\rangle$ unchanged.
    Since \(Q\) reaches \(u_i\) no later than $R$ (by Def.~\ref{def:multisource}), the resulting path $P'$ is a valid temporal path and the arrival times of nodes $u_{i+1},\ldots u_p$ through $P'$ remain the same as they were in $P$ (and, thus, the desired property is preserved for those nodes).
    This process can be repeated until every prefix-subpath is an MSEAP.
\qed
\end{proof}

\begin{lemma}\label{lem:earliest_multisource_correctness}
    Algorithm~\ref{alg:earliest_multisource_stream} computes the MSEAP times to each vertex, within time window $[t_\alpha, t_\omega]$. If time $\infty$ is returned for some vertex, then that vertex is unreachable from all $v \in V$ (starting no earlier than $\text{init}[v]$) within $[t_\alpha, t_\omega]$.
\end{lemma}

\begin{proof}
    For each \(v \in V\) let \(\tau[v]\) be the time calculated by the algorithm and \(\tau^*[v] \) be the respective MSEAP time (within $[t_\alpha, t_\omega]$). Since \(\tau[v]\) is initialized to $\text{init}[v]$, any update done by the algorithm calculates the arrival time of a valid temporal path, starting from some node $v$ at time $\text{init}[v]$. Thus, at any step of the algorithm it holds that \(\tau[v] \geq \tau^*[v]\), $\forall v \in V$.

    Fix an arbitrary vertex \(y\) with \(\tau^*[y]<\infty\). By Lemma~\ref{lem:earliest_multisource_prefix_property}, there exists an MSEAP $P=\langle v_{0}=x,v_{1},\dots ,v_{p}=y\rangle$ to $y$ s.t. every prefix-subpath $\langle v_0,v_1,v_2,\ldots,v_i \rangle$ of $P$ is an MSEAP to $v_i$. 
    Denote the \(i\)-th edge of $P$ by \(e_i = (v_{i-1},v_{i},t_{i},\lambda_{i}), i=1,2,\dots,p\). Thus, it holds that \(t_{i}\geq \tau^*[v_{i-1}]\) and \(\tau^*[v_{i}]=t_{i}+\lambda_{i}\). 
    We will show by induction on \(i\) that Algorithm \ref{alg:earliest_multisource_stream} computes \( \tau[v_i] = \tau^*[v_i]\) after processing \(e_i\).
    
    \emph{Base} \((i=0)\). By the prefix-subpath property of $P$, we have \(\tau ^*[x]=\text{init}[x]\). By the initialization of the algorithm, $\tau[x] \leq \text{init}[x]$. Thus, $\tau[x]=\tau^*[x]$.

    \emph{Inductive step} \((i \geq 1)\). Assume \(\tau[v_{i-1}]=\tau^*[v_{i-1}]\) after \(e_{i-1}\) is processed.
    When Algorithm \ref{alg:earliest_multisource_stream} scans \(e_{i}=(v_{i-1},v_{i},t_{i},\lambda_{i})\), it holds that $t_i \geq \tau[v_{i-1}]=\tau^*[v_{i-1}]$. Thus, \(\tau[v_{i}]\) is updated to \(t_{i}+\lambda_{i}=\tau^*[v_{i}]\) (unless it was already equal to that value). Since it always holds that \(\tau[v] \geq \tau^*[v]\), $\forall v \in V$, $\tau[v_i]$ will never be updated again after this step.

    From the above we infer that $\tau[v]= \infty$ can be returned for some $v\in V$ only if no MSEAP to $v$ exists. This implies that $v$ is unreachable from all $u \in V$ (starting no earlier than $\text{init}[u]$) within $[t_\alpha, t_\omega]$.
\qed
\end{proof}

We now present our EABP algorithm, which calls Alg.~\ref{alg:earliest_multisource_stream} twice, computing earliest-arrival path times in the first phase and EABP times in the second phase.

\begin{algorithm}[ht]
\caption{$\texttt{EABP$(G,x,[t_\alpha,t_\omega])$}$}
\begin{algorithmic}[1]

\Require{A temporal graph $G$, source vertex $x$, time window $[t_{\alpha}, t_{\omega}]$, a set of beer vertices $B$ and a list $T_i$ of time units for each of them.}
\Ensure{The EABP time from $x$ to every vertex $v \in V$ within $[t_\alpha, t_\omega]$.}

\State $\text{init}[x]\gets t_\alpha,\, \text{init}[v] \gets \infty$ \textbf{for all} $v \in V \setminus \{x\}$

\State $\tau[v] \leftarrow \texttt{MSEAP($G,[t_\alpha,t_\omega],\text{init}$)}$, for $v \in V$

\State $\text{init}'[v] \gets \infty$ \textbf{for all} $v \in V$

\For{{\textbf{all}} $b\in B$}{
    \State $beer\_ time \gets \min\{time \in T_b \mid time \geq \tau[b]\}$ \Comment{$\infty$ if it does not exist}

    \If{$beer\_ time \leq t_{\omega}$}
        \State $\text{init}'[b] \gets beer\_ time$
    \EndIf

}
\EndFor    
\State $\tau'[v] \leftarrow \texttt{MSEAP($G,[t_\alpha,t_\omega],\text{init}'$)}$, for $v \in V$

\State \Return{$\tau'[v]$ for each $v \in V$}
\end{algorithmic}
\label{alg:earliest_arrival_beer}
\end{algorithm}

\makeatletter
\newcommand{\starlabel}[1]{\protected@write\@auxout{}%
  {\string\newlabel{#1}{{\arabic{theorem}}{\thepage}}}}
\makeatother

\begin{theorem}
\makeatletter\edef\@currentlabel{\arabic{theorem}}\makeatother\label{thrm:earliest_beer_correctness}
    Algorithm~\ref{alg:earliest_arrival_beer} computes the EABP times from $x$ to all $v \in V$ in time $\bigO(n+M+k\log T)$.
    
\end{theorem}

\begin{proof}
    Algorithm~\ref{alg:earliest_arrival_beer} scans all edges twice. A binary search is done for each of the $k$ lists $T_b$. Hence, Algorithm~\ref{alg:earliest_arrival_beer} runs in time $\bigO(n+M+k\log T)$.

    The first call to Algorithm~\ref{alg:earliest_multisource_stream} computes the (single-source) earliest-arrival path times from $x$ to the beer vertices (within $[t_\alpha,t_\omega]$). Thus, the binary search on $T_b$ computes for each $b \in B$ the earliest departure time from $b$ s.t. $b$ is used as the beer vertex of a beer path from $x$ to any other node (within $[t_\alpha,t_\omega]$). Denote this time as $t[b]$, $b \in B$. Observe that, for any node $v$, the EABP time from $x$ to $v$ is the earliest time in which $v$ is reachable from some $b \in B$, starting from $b$ no earlier than $t[b]$. This coincides with the MSEAP time to $v$ for $\text{init}[b]=t[b]$, $b \in B$, and $\text{init}[u] = \infty$, $u \in V\setminus B$. Hence, the second call to Algorithm~\ref{alg:earliest_multisource_stream} correctly computes the EABP times from $x$ to all nodes (within $[t_\alpha,t_\omega]$), by Lemma~\ref{lem:earliest_multisource_correctness}.
\qed
\end{proof}

\subsubsection{A faster earliest-arrival beer path algorithm with adjacency list representation}

 We present an alternative algorithm to be used as the MSEAP subroutine of Algorithm~\ref{alg:earliest_arrival_beer}, in the place of Algorithm~\ref{alg:earliest_multisource_stream}. Although faster for graphs with no dominated edges, this algorithm requires an adjacency list instead of an edge stream, which may be costly to produce (see e.g.~\cite{Caro_temporal_logs}). The algorithm is similar in spirit to Dijkstra's algorithm, using a min-priority queue for the vertices and their arrival times, while utilizing binary search on parallel edges to achieve a speedup. Note that the input graph of Alg.~\ref{alg:earliest_arrival_beer} can be either in edge stream or in adjacency list representation, depending on which MSEAP subroutine is used.

\begin{algorithm}[ht]
\caption{$\texttt{MSEAP\textunderscore alternative$(G,[t_\alpha,t_\omega],\text{init})$}$}
\label{alg:earliest_multisource_dijkstra}
\begin{algorithmic}[1]
\Require{A temporal graph $G$ in adjacency list
 representation, a time window $[t_{\alpha}, t_{\omega}]$ and initial times $\text{init}[v]$, such that $t_\alpha \leq \text{init}[v] \leq t_\omega$ or $\text{init}[v] = \infty$, for all $v \in V$.}
\Ensure{The MSEAP times within $[t_\alpha, t_\omega]$ to every vertex, assuming we can start at time $\text{init}[v]$ from each $v \in V$.}

\State $\tau[v] \gets \text{init}[v]$ \textbf{for all} $v \in V$

\State $PQ \gets \emptyset$ \Comment{a min-priority queue}
\State $PQ.enqueue (\tau[v],v)$ \textbf{for all} $v \in V$
\Comment{Add $v$ to priority queue with $\tau[v]$ priority}

\While{$PQ \neq \emptyset$}
    \State $(\tau[u],u) \gets PQ.extractMin()$
    
    \For{$v \in N_{out}(u)$}
        \State Let $(u,v,t,\lambda) \in E(u,v)$ be a temporal edge with minimum $t$ s.t. $t \geq \tau[u]$ 
        
        \If{$t + \lambda \leq t_{\omega}$ \textnormal{\textbf{and} $t + \lambda < \tau[v]$}}
            \State $\tau[v] \gets t + \lambda$\;
             \State   $PQ.decreaseKey(v,\tau[v])$
        \EndIf
    \EndFor
\EndWhile

\State \Return{$\tau[v]$ for each $v \in V$}
\end{algorithmic}
\end{algorithm}

\begin{theorem}
\makeatletter\edef\@currentlabel{\arabic{theorem}}\makeatother\label{thrm:earliest_beer_correctness_alternative}
    For graphs with no dominated edges, Algorithm~\ref{alg:earliest_arrival_beer} computes the EABP times from $x$ to all $v \in V$ in time $\bigO(m\log \pi+n\log n + k \log T)$, if Algorithm~\ref{alg:earliest_multisource_dijkstra} is used as its MSEAP subroutine.
\end{theorem}

\begin{proof}
    Algorithm~\ref{alg:earliest_multisource_dijkstra} uses a binary search\footnote{Binary search is possible due to the assumption that there are no dominated edges in the graph.} on each non-empty set $E(u,v)$ $(u,v \in V)$; this requires total time $\bigO(m \log \pi)$. If a Fibonacci heap is used as priority queue, $\bigO(\log n)$ time is required for each extract-min operation, while decrease-key operations are done in constant time. Thus, Algorithm~\ref{alg:earliest_multisource_dijkstra} runs in time $\bigO(m\log \pi+n\log n)$. Recall that Algorithm~\ref{alg:earliest_arrival_beer} calls the MSEAP subroutine twice, while also doing $\bigO(k)$ binary searches on lists of size $\bigO(T)$. We infer that Algorithm~\ref{alg:earliest_arrival_beer} runs in time $\bigO(m\log \pi+n\log n + k \log T)$, if Algorithm~\ref{alg:earliest_multisource_dijkstra} is used as the MSEAP subroutine.

    We now prove the correctness of Algorithm~\ref{alg:earliest_multisource_dijkstra} by induction. We only describe the induction briefly, since it is similar to the proof of correctness of Dijkstra's classic algorithm. For each \(v \in V\) let \(\tau[v]\) be the time calculated by the algorithm and \(\tau^*[v] \) be the respective MSEAP time within $[t_\alpha, t_\omega]$. By definition of MSEAP, for the node with minimum $\text{init}[u]$ it is $\tau^*[u]=\text{init}[u]$; thus, the first node extracted from the queue has \(\tau[u]=\tau^*[u]\). For each $v \in N_{out}(u)$, $\tau[v]$ is then updated through the $u$--$v$ edge that arrives at $v$ at the earliest possible time. Assuming the graph has no dominated edges, this is an edge with minimum $t$ s.t. $t \geq \tau[u]$. Inductively, any $u\in V$ extracted from the queue has \(\tau[u]=\tau^*[u]\). This happens because paths that only use processed nodes have been considered in previous steps and it is impossible to arrive at $u$ earlier than $\tau[u]$ by traversing an unprocessed node $v \in V$ with $\tau[v] \geq \tau[u]$. Regarding the update of $\tau[v]$ for $v \in N_{out}(u)$, the same things hold as in the base of the induction, thus completing the proof.

    The correctness of Algorithm~\ref{alg:earliest_arrival_beer} using Algorithm~\ref{alg:earliest_multisource_dijkstra} as the MSEAP subroutine is then identical to the proof of Theorem~\ref{thrm:earliest_beer_correctness}.
\qed
\end{proof}

 The $n\log n$ term in the complexity stems from the usage of a Fibonacci heap as min-priority queue. The $m\log \pi$ term stems from binary search on parallel edges, which is possible by the assumption that $G$ contains no dominated edges. Note that $M$ is bounded by $m \pi$ and may be arbitrarily larger than $m$, which renders $\bigO(m \log \pi)$ a significant improvement over $\bigO(M)$.


\subsection{Algorithms for latest-departure beer paths}\label{subsec:latest_beer}

We propose two LDBP algorithms (from all nodes to some node $y$), based on our EABP algorithms with modifications for \emph{multiple-target} latest departure paths (MTLDP).
The main idea is to reverse the direction of time, starting from~$y$ in~$t_\omega$ and traversing temporal edges in reverse, while storing the latest time in which we can depart from each node to reach $y$. In these algorithms, node times are initialized to $-\infty$ and the edge stream is scanned in reverse.
Due to their similarity to the algorithms of Section~\ref{subsec:earliest_beer}, we will mainly focus on their differences compared to our EABP algorithms. We also include the algorithms' pseudocode for the sake of completeness, although the proposed modifications are fairly simple.
We first give the following definition, which is analogous to Def.~\ref{def:multisource} (MSEAP).

\begin{definition}[MTLDP]\label{def:multitarget}
   Given a temporal graph $G=(V,E)$ and final times $\text{fin}[v]\in \mathbb{R}^+ \cup \{-\infty\}$, $v\in V$, we say that a temporal $x$--$y$ path arriving at $y$ no later than $\text{fin}[y]$ is a \emph{multiple-target latest-departure path} from $x$ if its starting time from $x$ is maximum among all temporal paths $x-v$, $v\in V$, arriving at $v$ no later than $\text{fin}[v]$. We call this starting time the \emph{multiple-target latest-departure path time} of $x$.
\end{definition}

Intuitively, a MTLDP is a path that allows us to leave from a given vertex as late as possible, in order to reach some target within its respective given time limit. A practical application would be to buy something from a shop as late as possible, given the varying closing times of each shop.

We present two MTLDP algorithms (Algorithms~\ref{alg:latest_multisource_stream} and~\ref{alg:latest_multisource_dijkstra}), to be used as subroutines for computing LDBPs. The first one uses an edge stream (analogous to Alg.~\ref{alg:earliest_multisource_stream}) and the second one uses an adjacency list (analogous to Alg.~\ref{alg:earliest_multisource_dijkstra}). The time complexity of each algorithm is the same as in Section~\ref{subsec:earliest_beer} and the correctness proofs are analogous.
For Algorithm~\ref{alg:latest_multisource_stream}, the main difference compared to Alg.~\ref{alg:earliest_multisource_stream} is that the edge stream is scanned in reverse (i.e., by decreasing $t$-value). Each edge $(u,v,t,\lambda)$ updates $\tau [u]$ to $t$ if $\tau [u] < t$ and $t \geq t_a$ and $t+\lambda \leq \tau [v]$.

\begin{algorithm}[ht]
\caption{$\texttt{MTLDP$(G,[t_\alpha,t_\omega],\text{fin})$}$}
\label{alg:latest_multisource_stream}
\begin{algorithmic}[1]

\Require{A temporal graph $G$ in its reverse edge stream
 representation, a time window $[t_{\alpha}, t_{\omega}]$ and final times $\text{fin}[v]$, such that $t_\alpha \leq \text{fin}[v] \leq t_\omega$ or $\text{fin}[v] = -\infty$, for all $v \in V$.}
\Ensure{The MTLDP times within $[t_\alpha, t_\omega]$ from every vertex, assuming we can arrive at time $\text{fin}[v]$ at each $v \in V$.}

\State $\tau[v] \gets \text{fin}[v]$ \textbf{for all} $v \in V$

\For{$e=(u,v,t,\lambda)$ in the reverse edge stream}{
    \If{$t \geq t_{\alpha}$ \textnormal{\textbf{and}} $t+\lambda\leq \tau[v]$}
       \State $\tau[u] \gets \max\{\tau[u], t\}$
    \EndIf
    
}
\EndFor

\State \Return{$\tau[v]$ for each $v \in V$}
\end{algorithmic}
\end{algorithm}

 For Algorithm~\ref{alg:latest_multisource_dijkstra}, a max-priority queue is used for the MTLDP times and the binary search on parallel edges considers their $(t+\lambda)$-value.\footnote{Note that non-dominated edges sorted by $t$-value are also sorted by $(t+\lambda)$-value (by definition). Thus, this binary search is possible with no additional precomputation.} 
For each node $v$ extracted from the queue we consider its set $N_{in}(v)$ and the updates to the queue are done through \emph{increase-key} operations.

\begin{algorithm}[ht]
\caption{$\texttt{MTLDP\textunderscore alternative$(G,[t_\alpha,t_\omega],\text{fin})$}$}
\label{alg:latest_multisource_dijkstra}
\begin{algorithmic}[1]
\Require{A temporal graph $G$ in adjacency list
 representation, a time window $[t_{\alpha}, t_{\omega}]$ and final times $\text{fin}[v]$, such that $t_\alpha \leq \text{fin}[v] \leq t_\omega$ or $\text{fin}[v] = -\infty$, for all $v \in V$.}
\Ensure{The MTLDP times within $[t_\alpha, t_\omega]$ from every vertex, assuming we can arrive at time $\text{fin}[v]$ at each $v \in V$.}

\State $\tau[v] \gets \text{fin}[v]$ \textbf{for all} $v \in V$

\State $PQ \gets \emptyset$ \Comment{a max-priority queue}
\State $PQ.enqueue (\tau[v],v)$ \textbf{for all} $v \in V$
\Comment{Add $v$ to priority queue with $\tau[v]$ priority}

\While{$PQ \neq \emptyset$}
    \State $(\tau[v],v) \gets PQ.extractMax()$
    
    \For{$u \in N_{in}(v)$}
        \State Let $(u,v,t,\lambda) \in E(u,v)$ be a temporal edge with maximum $t$ s.t. $t + \lambda \leq \tau[v]$ 
        
        \Comment{binary search on $t+\lambda$ of $e \in E(u,v)$ }
        \If{$t \geq t_\alpha$ \textnormal{\textbf{and} $t > \tau[u]$}}
            \State $\tau[u] \gets t $\;
             \State  $PQ.increaseKey(u,\tau[u])$
        \EndIf  
    \EndFor
\EndWhile

\State \Return{$\tau[v]$ for each $v \in V$}
\end{algorithmic}
\end{algorithm}

We now present our main LDBP algorithm (analogous to Alg.~\ref{alg:earliest_arrival_beer}), which may use either Alg.~\ref{alg:latest_multisource_stream} or Alg.~\ref{alg:latest_multisource_dijkstra} as MTLDP subroutine. As discussed in Section~\ref{subsec:earliest_beer}, the latter is a faster alternative for graphs with no dominated edges, but it requires random access to the edges of the graph. 

Algorithm~\ref{alg:latest_departure_beer} computes latest-departure times in the first phase and LDBP times in the second phase, with a binary search in each list $T_b$ between the two phases. Note that the latest-departure times of all nodes except the target node are initialized to $-\infty$.
The input graph of Alg.~\ref{alg:latest_departure_beer} can be either in reverse edge stream or adjacency list representation, depending on which MTLDP subroutine is used.

\begin{algorithm}[ht]
\caption{$\texttt{LDBP$(G,y,[t_\alpha,t_\omega])$}$}
\begin{algorithmic}[1]

\Require{A temporal graph $G$, target vertex $y$, a time window $[t_{\alpha}, t_{\omega}]$, a set of beer vertices $B$ and a list $T_i$ of time units for each of them.}
\Ensure{The LDBP time to $y$ from every vertex $v \in V$ within $[t_\alpha, t_\omega]$.}

\State $\text{fin}[y]\gets t_\omega,\, \text{fin}[v] \gets -\infty$ \textbf{for all} $v \in V \setminus \{y\}$

\State $\tau[v] \leftarrow \texttt{MTLDP($G,[t_\alpha,t_\omega],\text{fin}$)}$, for $v \in V$

\State $\text{fin}'[v] \gets -\infty$ \textbf{for all} $v \in V$

\For{{\textbf{all}} $b\in B$}{
    \State $beer\_ time \gets \max\{time \in T_b \mid time \leq \tau[b]\}$
    \Comment{$-\infty$ if it does not exist}

    \If{$beer\_ time \geq t_{\alpha}$}
        \State $\text{fin}'[b] \gets beer\_ time$
    \EndIf

}
\EndFor    
\State $\tau'[v] \leftarrow \texttt{MTLDP($G,[t_\alpha,t_\omega],\text{fin}'$)}$, for $v \in V$

\State \Return{$\tau'[v]$ for each $v \in V$}
\end{algorithmic}
\label{alg:latest_departure_beer}
\end{algorithm}

\subsection{An algorithm for fastest beer paths }

A significant difference for FBP compared to EABP/LDBP is that no property similar to Lemma~\ref{lem:earliest_multisource_prefix_property} holds, meaning that prefix-subpaths of an FBP are not necessarily fastest paths. Hence, computing a fastest path from the source to a beer vertex is not useful as an intermediate step for computing FBPs. The following lemma is the key to obtaining an FBP algorithm with time complexity comparable to that of our EABP/LDBP algorithms in the previous subsections (excluding some log factors), despite the aforementioned complication.

\begin{lemma}
\label{lem:fastest}
    If there is a temporal $x$--$y$ beer path within time interval $[t_\alpha, t_\omega]$, then there is an $x$--$y$ FBP $\langle v_0=x,v_1,\ldots,v_j=b,\ldots,v_p=y\rangle$ within $[t_\alpha, t_\omega]$, traversing $b \in B$ at a time $t \in T_b$, such that its subpaths $\langle v_0=x,v_1,\ldots,v_j=b\rangle$ and $\langle v_j=b,v_{j+1},\ldots,v_p=y\rangle$ are non-dominated paths.
\end{lemma}

\begin{proof}
    Let there be a temporal beer path $P=\langle u_0=x,u_1,\ldots,u_{j'}=b,\ldots,u_{p'}=y\rangle$ within $[t_\alpha, t_\omega]$, traversing $b \in B$ at a time $t \in T_b$. Define the subpaths of $P$: $Q=\langle u_0=x,u_1,\ldots,u_{j'}=b\rangle$, $R=\langle u_{j'}=b,u_{j'+1},\ldots,u_{p'}=y\rangle$. If $Q$ is a dominated path, then there is some other path $Q'=\langle w_0=x,w_1,\ldots,w_f=b\rangle$ such that $\text{start}(Q') \geq \text{start}(Q)$ and $\text{end}(Q') \leq \text{end}(Q)$, with at least one of the inequalities being strict. Concatenate $Q'$ with $R$ to obtain a $x$--$y$ path $P'$. By the first inequality we obtain $\text{dura}(P') \leq \text{dura}(P)$. By the second inequality we infer that $P'$ also traverses $b$ at time $t$ and that $P'$ is a valid temporal path. Thus, $P'$ is a beer path with smaller or equal duration than $P$ and is also within $[t_\alpha, t_\omega]$. Applying a similar argument for $R$ completes the proof. 
\qed
\end{proof}

Based on Lemma~\ref{lem:fastest}, we propose a two-phase FBP algorithm (from source $x$ to all nodes), utilizing Algorithm~\ref{alg:non_dominated_paths} to compute non-dominated paths. In the first phase, all non-dominated paths (excluding ties) from $x$ to each $b \in B$ are computed. Similar to Section~\ref{subsec:earliest_beer}, we find the minimum eligible active time for the respective beer vertex for each of these paths. 
For the second phase, a dummy vertex is used as source, connected with each $b \in B$ by (multiple) dummy edges with starting and traversal times dictated by the values previously computed for each path. Algorithm~\ref{alg:non_dominated_paths} is then used again for the new graph, computing eligible non-dominated paths from the dummy vertex to all vertices. 

\begin{algorithm}[ht]
\caption{$\texttt{FBP$(G,x,[t_\alpha,t_\omega])$}$}
\label{alg:fastest_beer}
\begin{algorithmic}[1]
\Require{A temporal graph $G$ in edge stream representation, source vertex $x$, time interval $[t_\alpha, t_\omega]$, a set of beer vertices $B$ and a list $T_i$ of time units for each of them.}
\Ensure{The FBP duration from $x$ to every vertex $v \in V$ within $[t_\alpha, t_\omega]$.}


\State $L_b \leftarrow \texttt{Non-dom\textunderscore paths($G,x,[t_\alpha,t_\omega]$)}$, for $b \in B$
\Comment{call Alg~\ref{alg:non_dominated_paths} for non-dom. paths}

\For{$b\in B$}

\For{$(s,a) \in L_b$}
    \State $a' \gets \min\{time \in T_b \mid time \geq a\}$  \Comment{$\infty$ if it does not exist}

    \If{$a' \leq t_{\omega}$}
        \State Replace $(s,a)$ with $(s,a')$ in $L_b$
        
    \Else
        \State Remove $(s,a)$ from $L_b$ \Comment{impossible to get beer}
    \EndIf
    \EndFor
\EndFor

\State Add dummy node $\epsilon$ to the graph. Add an edge $(\epsilon, b, s, a-s)$ for each $(s,a) \in L_b$ for each $b \in B$; Let $G'=(V',E')$ be the modified graph\;

\State $L_v' \leftarrow \texttt{Non-dom\textunderscore paths($G',\epsilon,[t_\alpha,t_\omega]$)}$, for $v \in V$
\Comment{Alg~\ref{alg:non_dominated_paths} in modified graph}

\For{$v \in V$}
\State $f[v] \leftarrow a-s$, where $(s,a)\in L_v'$ is the tuple with minimum difference
\EndFor
\State \Return{$f[v]$ for each $v \in V$} 
\end{algorithmic}
\end{algorithm}

We remark that the last for-loop of Alg.~\ref{alg:fastest_beer} can be incorporated into the second call of Alg.~\ref{alg:non_dominated_paths} without increasing its time complexity, by storing the tuple with minimum difference. We only present the pseudocode in this manner to help readability.

Alg.~\ref{alg:non_dominated_paths} runs in $\bigO(n + M \log c)$ time (cf.~\cite{Path_Wu}), where $c = \min\{d_{in}, d_{out}\}$ (see Section~\ref{auxiliary_section_dom_path_algo} for details). Hence, our FBP algorithm (Alg.~\ref{alg:fastest_beer}) runs in time $\bigO(n+(M+kc)\log d_{in}+kc\log T)$, since $\bigO(kc)$ temporal edges are introduced for the second phase and $d_{in}$ becomes (at most) double its initial value. The term $kc\log T$ stems from binary searches over beer times. The correctness of Alg.~\ref{alg:fastest_beer} follows directly from Lemma~\ref{lem:fastest}.

\subsection{An algorithm for shortest beer paths}

SBPs are significantly more complicated, in part because, somewhat unexpectedly, Lemma~\ref{lem:fastest} does not hold verbatim for \emph{distance-wise} non-dominated paths. To overcome this, we define the following alternative domination criterion.

\begin{definition}[\textbf{Inverse-distance-wise dominated path}]
\label{def:inverse_dom_path}
Let $P$ be a temporal $u$--$v$ path. We say that $P$ is \emph{inverse-distance-wise dominated} if there exists another temporal $u$--$v$ path $P'$ such that $\text{dist}(P') \le \text{dist}(P)$ and $\text{start}(P') \ge \text{start}(P)$, with at least one of the inequalities being strict.
\end{definition}

This domination criterion allows us to prove the following lemma, which is analogous to Lemma~\ref{lem:fastest} and serves as the backbone of our SBP algorithm.

\begin{lemma}
\label{lem:shortest}
    If there is a temporal $x$--$y$ beer path within time interval $[t_\alpha, t_\omega]$, then there is a $x$--$y$ SBP $\langle v_0=x,v_1,\ldots,v_j=b,\ldots,v_p=y\rangle$ within $[t_\alpha, t_\omega]$, traversing $b \in B$ at a time $t \in T_b$, such that $\langle v_0=x,v_1,\ldots,v_j=b\rangle$ is a distance-wise non-dominated path (within $[t_\alpha, t_\omega]$) and $\langle v_j=b,v_{j+1},\ldots,v_p=y\rangle$ is an inverse-distance-wise non-dominated path (within $[t_\alpha, t_\omega]$).
\end{lemma}

\begin{proof}
    Let there be a temporal beer path \(P=\langle u_{0}=x,u_{1},\dots ,u_{j'}=b,\dots ,u_{p'}=y\rangle\) within \([t_{\alpha},t_{\omega}]\), traversing \(b\in B\) at some time \(t\in T_{b}\).  
    Define the sub-paths of \(P\):   \(Q=\langle u_{0}=x,u_{1},\dots ,u_{j'}=b\rangle\) and  \(R=\langle u_{j'}=b,u_{j'+1},\dots ,u_{p'}=y\rangle\). If \(Q\) is distance-wise dominated (within \([t_{\alpha},t_{\omega}]\)), there exists a path  
    \(Q'=\langle w_{0}=x,w_{1},\dots ,w_{f}=b\rangle\) (within \([t_{\alpha},t_{\omega}]\)), such that  
    \(\text{dist}(Q')\le \text{dist}(Q)\) and \(\text{end}(Q')\le \text{end}(Q)\), with at least one of the two inequalities strict.  
    Concatenate \(Q'\) with \(R\) to obtain a path \(P'\) from \(x\) to \(y\).  Because of the first inequality we get \(\text{dist}(P') \le \text{dist}(P)\). 
    By the second inequality we infer that \(P'\) also traverses \(b\) at time $t$ and is a valid temporal path. Hence, \(P'\) is a beer path within \([t_{\alpha},t_{\omega}]\) whose total distance is smaller than or equal to that of \(P\). 
    
    If \(R\) is inverse-distance-wise dominated (within \([t_{\alpha},t_{\omega}]\)), there is a path \(R'=\langle z_{0}=b,z_{1},\dots ,z_{g}=y\rangle\) (within \([t_{\alpha},t_{\omega}]\)) such that  \(\text{dist}(R')\le \text{dist}(R)\) and \(\text{start}(R')\ge \text{start}(R)\), again with at least one inequality being strict. Concatenate \(Q'\) with \(R'\) to obtain a $x$--$y$ path $P''$. By the first inequality we obtain \(\text{dist}(P'') \le \text{dist}(P')\). By the second inequality, we obtain that $P''$ is a valid temporal beer path. Therefore, \(P''\) is a beer path within \([t_{\alpha},t_{\omega}]\) whose total distance is smaller than or equal to that of \(P'\).
\qed
\end{proof}




We now propose a modification of Alg.~\ref{alg:distance_non_dominated_paths} to compute inverse-distance-wise non-dominated paths from all vertices to some vertex $y$, by scanning the edge stream in reverse. 
Algorithm~\ref{alg:inverse_distance_non_dominated_paths} computes all inverse-distance-wise non-dominated paths (excl. ties) from all vertices to some given target-vertex, within a given time interval. In contrast to Alg.~\ref{alg:distance_non_dominated_paths}, it uses a \emph{reverse} edge stream (similar to our LDBP algorithm in Section~\ref{subsec:latest_beer}) and stores a tuple $(d,s)$ for each path $P$, where $d=dist(P)$ and $s=start(P)$. To understand why a reverse edge stream is necessary and why the algorithm runs for a specific target (instead of source), we refer the reader to the respective proof of correctness (Theorem~\ref{thrm:inv_dist_correctness}). 

\begin{algorithm}[ht]
\caption{$\texttt{Inv\textunderscore dist\textunderscore non-dom\textunderscore paths($G,y,[t_\alpha,t_\omega]$)}$}
\label{alg:inverse_distance_non_dominated_paths}
\begin{algorithmic}[1]
\Require{A temporal graph $G = (V, E)$ in reverse edge stream representation, target vertex $y$, time interval $[t_\alpha, t_\omega]$.}


\Ensure{For each $v \in V$, the list $I_v$ of inverse-distance-wise non-dominated $(d, s)$ pairs (from $v$ to $y$) within \([t_{\alpha},t_{\omega}]\).}
\For{$v \in V$}
    \State Create a sorted list $I_v$, where an element of $I_v$ is a pair $(d, s)$ in which $d$ is the distance of a path $P$ from $v$ to $y$ and $s$ is the time in which path $P$ starts from $v$ and is used as the key for ordering in $I_v$; initially, $I_v$ is empty;\;
    \EndFor
\For{$e = (u, v, t, \lambda)$ in the reverse edge stream}
    \If{$t \ge t_\alpha$ and $t + \lambda \le t_\omega$}
        \If{$v = y$ and $(0, t+ \lambda) \not\in I_y$}
            \State Insert $(0, t+ \lambda)$ into $I_y$\;
        \EndIf
        \State Let $(d', s')$ be the element in $I_v$ such that $s' = \min\{s \mid (d, s) \in I_v,\ s \ge t+\lambda \}$\;
        \If{$\exists\, (d,s) \in I_u$ \textnormal{such that} $s = t$ and $d > d' + \lambda$}
                \State Replace $(d,s)$ with $(d' + \lambda, t)$\;
            
        \Else
                \State Insert $(d' + \lambda, t)$ into $I_u$\;
        \EndIf
        \State Remove inverse-distance-wise dominated elements in $I_u$\;
    \EndIf
\EndFor
\State \Return{$I_v$ for all $v \in V$}\;
\end{algorithmic}
\end{algorithm}

Algorithm~\ref{alg:inverse_distance_non_dominated_paths} runs in time $\bigO(n+M \log d_{out})$. Observe that $d_{out}$ is a bound for the amount of inverse-distance-wise non-dominated pairs $(d,s)$ (distance, start) between any two vertices and, thus, also a bound for the size of each list $I_v$. 

\begin{theorem}
\label{thrm:inv_dist_correctness}
    Algorithm~\ref{alg:inverse_distance_non_dominated_paths} computes all inverse-distance-wise non-dominated $(d,s)$ pairs from each $v \in V$ to $y$ within $[t_\alpha,t_\omega]$, in time $\bigO(n+M \log d_{out})$.
\end{theorem}

\begin{proof}
    Let \(P=\langle v_0=v,v_1,\dots ,v_p=y\rangle\) be an inverse-distance-wise non-dominated path from $v$ to $y$ within $[t_\alpha,t_\omega]$. For $i\in \{0,\ldots,p-1\}$, let $s_i$ be the time in which $P$ departs from $v_i$ and $d_i$ be the distance between $v_i$ and $y$ in~$P$. We will prove that $(d_0,s_0)$ is contained in the list $I_v$ returned by Algorithm~\ref{alg:inverse_distance_non_dominated_paths}.
    Since we use a reverse edge stream, the edges of $P$ are processed in reverse order (starting from~$y$). We use induction on the edges of $P$.

    \emph{Base}. Let $e_p=(v_{p-1},y,t_p,\lambda_p)$ be the last edge of $P$. When $e_p$ is processed, the algorithm inserts $(0,t_p+\lambda_p)$ to $I_y$. Then, $(\lambda_p,t_p)$ will be inserted to $I_{v_{p-1}}$ unless that list already contains a pair that dominates $(\lambda_p,t_p)$ or one that is identical to it. In either case, $I_{v_{p-1}}$ will end up containing an element $(d,s)$ with $d\leq \lambda_p = d_{p-1}$ and $s\geq t_p = s_{p-1}$. Additionally, $I_{v_{p-1}}$ contains no dominated elements (because the algorithm removes them at the end of each loop).

    \emph{Inductive step}. Assume that $I_{v_i}$ has no dominated elements and contains an element $(d,s)$ with $d\leq d_i$ and $s \geq s_i$ (induction hypothesis). Let $e_i=(v_{i-1},v_i,t_i,\lambda_i)$ be the $i$-th edge of $P$ and note that $s_i\geq t_i +\lambda_i$ because $P$ is a valid temporal path. When $e_i$ is processed, the algorithm finds the best fitting pair in $I_{v_i}$, i.e. $(d^*,s^*)$ such that $s^*$ is the minimum possible start larger than $t_i +\lambda_i$. By the induction hypothesis, such an $s^*$ exists and $d^* \leq d_i$. Thus, after $e_i$ is processed, $I_{v_{i-1}}$ will contain an element $(d,s)$ such that $d \leq d_i + \lambda_i = d_{i-1}$ and $s \geq t_i = s_{i-1}$. Also, $I_{v_{i-1}}$ will contain no dominated elements.

    From the above induction we infer that $I_v$ contains an element $(d,s)$ with $d \leq d_0$ and $s \geq s_0$. By assumption, $P$ is an inverse-distance-wise non-dominated path within $[t_\alpha,t_\omega]$, which implies that neither of the above inequalities can be strict. Hence, $I_v$ contains $(d_0,s_0)$, which finishes the correctness proof.

    For the running time of the algorithm, observe that the amount of inverse-distance-wise non-dominated pairs from $x\in V$ to $y\in V$ is bounded by the out-degree of $x$. This implies that the size of any list $I_v$ is bounded by $d_{out}$ and, thus, search/update operations in $I_v$ take $\bigO(\log d_{out})$ time. Note that for each list $I_v$, at most one element is added for each of the outgoing edges of $v$. Thus, the total time of removing dominated elements from $I_v$ is bounded by the out-degree of $v$. For all lists $I_v,\ v\in V$, this adds up to $\bigO(M)$ time. From all the above combined with the fact that Algorithm~\ref{alg:inverse_distance_non_dominated_paths} processes each temporal edge once, we obtain that it runs in $\bigO(n+M\log d_{out})$ time.    
\qed
\end{proof}

We now propose an SBP algorithm, utilizing Algorithms~\ref{alg:distance_non_dominated_paths} and \ref{alg:inverse_distance_non_dominated_paths} to compute (inverse) distance-wise non-dominated paths. Unlike most of our algorithms, Algorithm~\ref{alg:shortest_beer} only computes an SBP between two given vertices and it requires both the regular and the reverse edge streams as input; this is necessary because it uses two different domination criteria (see Lemma~\ref{lem:shortest}), with each one imposing different restrictions via its respective algorithm (i.e., Alg.~\ref{alg:distance_non_dominated_paths} is one-to-all and Alg.~\ref{alg:inverse_distance_non_dominated_paths} is all-to-one).

\begin{algorithm}[ht]
\caption{$\texttt{SBP$(G,x,y,[t_\alpha,t_\omega])$}$}
\label{alg:shortest_beer}
\begin{algorithmic}[1]
\Require{A temporal graph $G = (V, E)$ in edge stream representation (both regular and reverse), source vertex $x$, target vertex $y$, time interval $[t_\alpha, t_\omega]$, a set of beer vertices $B$ and a list $T_i$ of time units for each of them.}
\Ensure{The SBP distance from $x$ to $y$ within $[t_\alpha, t_\omega]$.}

\State $dist \leftarrow \infty$
\State $L_b \leftarrow \texttt{Dist\_Non-dom\_paths($G,x,[t_\alpha,t_\omega]$)}$, for $b \in B$  
\Comment{Alg~\ref{alg:distance_non_dominated_paths} from $x$ to $b\in B$}
\State $I_b \leftarrow \texttt{Inv\textunderscore dist\textunderscore Non-dom\textunderscore paths($G,y,[t_\alpha,t_\omega]$)}$, for $b \in B$ 
\Comment{Alg~\ref{alg:inverse_distance_non_dominated_paths} from $b \in B$ to $y$}

\For{$b\in B$}
    \For{$(d,a) \in L_b$}
        \State $a' \gets \min\{time \in T_b \mid time \geq a\}$ \Comment{$\infty$ if it does not exist}

        \If{$a' \leq t_{\omega}$}
            \State Replace $(d,a)$ with $(d,a')$ in $L_b$
        
        \Else 
            \State Remove $(d,a)$ from $L_b$
        \Comment{impossible to get beer}
        \EndIf
\EndFor
    \For{$(d,a) \in L_b$}
        \State Find $(d',s)$ in $I_b$ with minimum $s$ such that $s \geq a$
        \Comment{best fitting path}

        \State $\text{dist} \leftarrow \min\{d+d' , \text{dist}\}$
    \EndFor
\EndFor
\State \Return{dist}
\end{algorithmic}
\end{algorithm}

Alg.~\ref{alg:shortest_beer} calls Alg.~\ref{alg:distance_non_dominated_paths} and~\ref{alg:inverse_distance_non_dominated_paths} once each. The size of each $L_v$ is bounded by $d_{in}$ and the size of each~$I_v$ by $d_{out}$. Hence, Alg.~\ref{alg:shortest_beer} runs in time $\bigO(n+M \log (d_{max}) + k d_{in} (\log T + \log d_{out}))$.
The correctness of Alg.~\ref{alg:shortest_beer} follows from Lemma~\ref{lem:shortest} and Theorem~\ref{thrm:inv_dist_correctness}.

\section{Temporal beer path queries with preprocessing}\label{sec:preprocessing}

In this section we present preprocessing techniques that help improve the asymptotic complexity of beer path queries. This is useful for applications where multiple queries need to be answered, each with different source and target vertices, different time intervals and even different active beer shops. A natural application would be to quickly recompute a beer path assuming some beer shops just closed (or some new ones opened), without computing everything from scratch.

We assume that each query contains as input a source node $x$, a target node $y$, a time interval $[t_\alpha,t_\omega]$ and a Boolean array $A=[a_1,\ldots,a_k]$, with $a_i = \texttt{true}$ iff beer vertex $b_i$ is \emph{active}.
For simplicity, we assume that each beer vertex $b \in B$ is either active or inactive for the entire duration of a query; note that our methods can be easily extended to the case where active beer times are given as input in each query, by incorporating binary searches in lists $T_b$, as in Section~\ref{sec:algo}. We also assume that the set $B$ containing the $k$ beer vertices is known during preprocessing; these vertices may be active or inactive for different queries, but no vertex $v \notin B$ may become a beer vertex.

\subsection{EABP and LDBP with precomputation of non-dominated paths}

We rely on the observation that Lemma~\ref{lem:fastest} also holds for EABPs and LDBPs, with the same proof (although not for SBPs). This allows us to run Algorithm~\ref{alg:non_dominated_paths} to precompute non-dominated paths to and from beer vertices.
To compute all non-dominated paths to each beer vertex (starting from any vertex) and all non-dominated paths from each beer vertex (to any vertex), it suffices to run Algorithm~\ref{alg:non_dominated_paths} $\bigO(n)$ times. Recall that the amount of non-dominated paths (excl. ties) between any two vertices is bounded by $c= \min \{d_{in},d_{out}\}$. We infer that this preprocessing requires $\bigO(n^2 + nM \log c)$ time and $\bigO (nkc)$ space.

With this preprocessing, any EABP query can be answered as follows. We need a binary search for each $b_i \in B$ with $a_i = \texttt{true}$ to find the non-dominated $x$--$b_i$ path that starts as early as possible (but no earlier than $t_\alpha$). Then, for each $b_i \in B$  with $a_i = \texttt{true}$ we need a binary search to find the non-dominated $b_i$--$y$ path that starts as early as possible (but no earlier than the arrival time of the respective $x$--$b_i$ path). Concatenating these paths with the previous ones yields $k$ candidate EABPs, out of which we pick the best one. This process runs in $\bigO (k \log c)$ time. Identical arguments hold for LDBP queries.

\subsection{Preprocessing through graph transformation}

In this subsection we adopt the graph transformation approach proposed by~\cite{Path_Wu} as preprocessing. This technique converts a temporal graph $G = (V, E)$ into a static weighted graph $\widetilde{G} = (\widetilde{V}, \widetilde{E})$ that explicitly encodes temporal constraints. The key idea is to model the timing of events in the temporal graph by creating time-stamped copies of each node and connecting them with appropriately weighted edges. The transformed graph $\widetilde{G}$ contains $\bigO(M)$ vertices and edges and can be produced from $G$ in $\bigO(M\log d_{max})$ time.
For each $v \in V$, the transformation creates two sets of time-stamped nodes: $\widetilde{V}_{in}(v)$, which includes a node $(v, t)$ for every unique time $t$ at which an edge arrives at~$v$, and $\widetilde{V}_{out}(v)$, which includes a node $(v, t)$ for every time $t$ at which an edge departs from~$v$. These sets are sorted in descending time order.
We refer the reader to Section~\ref{auxiliary_section_graph_transformation} and Figure~\ref{fig:graph_transf} for more details regarding this transformation.

\subsubsection{Fastest beer paths on the transformed graph}

Let $S = \{(x, t) : (x, t) \in \widetilde{V}_{out}(x),$ $t_\alpha \leq t \leq t_\omega \}$ be the set of time-stamped copies of $x$, sorted in descending order of $t$. 
In order to compute fastest temporal paths from $x$ to every $v\in V$, it suffices to start a BFS from each $(x, t)\in S$ (in descending order of $t$), skipping nodes already explored in previous BFSs to avoid redundancy, since reaching $(v,t)$ from $(x,t_1)$ results in smaller duration compared to reaching it from $(x,t_2)$ with $t_2 < t_1$. The duration of each path is computed as the difference between the timestamp of the target and the timestamp of the source \cite{Path_Wu}. We will modify this technique to compute FBPs on $\widetilde{G}$.


For each vertex $(v,t)\in \widetilde{G}$, we keep track of a $3$-valued flag, indicating whether $(v,t)$ has been visited by a beer path (value $2$) or a regular path (value $1$) or not at all (value $0$), with higher values superseding smaller ones. We run BFSs from each $(x,t) \in S$ in descending  order of $t$ (as described above), setting the flag of $(x,t)$ to $1$. We only visit a vertex if its flag is smaller than the parent's flag; in this case, the child obtains the parent's flag. If the child is a copy of some $b_i \in B$ with $a_i = \texttt{true}$, we always update its flag to $2$. When the flag of a vertex becomes $2$, we calculate the duration of the respective beer path as the difference between the timestamp of that vertex and the timestamp of the source.

The above answers FBP queries on $\widetilde{G}$ in $\bigO(M)$ time, since each $v\in \widetilde{G}$ is visited at most twice. The correctness of the algorithm follows from the fact that $\widetilde{G}$ fully encodes all possible temporal walks in $G$, combined with arguments similar to the ones for regular fastest paths.

We remark that EABP and LDBP queries can also be answered in $\widetilde{G}$ in $\bigO(M)$ time with similar techniques. However, this is not particularly useful, since the algorithms proposed in Sections \ref{subsec:earliest_beer} and \ref{subsec:latest_beer} require no preprocessing and are just as fast (or even faster) under reasonable assumptions.




\subsubsection{Shortest beer paths on the transformed graph}\label{sec:SBP_transformed}

To compute shortest temporal paths from $x\in V$, it suffices to run Dijkstra's algorithm in $\widetilde{G}$ from a dummy vertex $x'$, connected via directed edges of weight 0 to every $(x, t) \in \widetilde{V}_{\text{out}}(x)$ s.t. $\ t_\alpha \leq t \leq t_\omega$, thus representing all valid departure times from $x$~\cite{Path_Wu}.

We modify Dijkstra's algorithm by storing a $3$-valued flag for each $(v,t)\in \widetilde{G}$ and updating it in the same manner described for FBP. This algorithm visits each $(v,t) \in \widetilde{G}$ at most twice, thus calculating SBP distances in $\bigO(M\log M)$ time. Although this SBP algorithm is not faster than Alg.~\ref{alg:shortest_beer} under reasonable assumptions, it computes SBPs from a source to \emph{all} other nodes. However, we can improve the running time with the following observation.

\begin{observation}
For a graph $G=(V,E)$ such that every temporal edge \(e=(u,v,t,\lambda)\in E\) has
\(\lambda>0\), the transformed graph $\widetilde{G}$ is a directed acyclic graph (DAG).
\end{observation}

Thus, when no edge with $\lambda=0$ exists in $G$, a one-to-all SBP query in the transformed graph $\widetilde{G}$ can be answered in \(\mathcal{O}(M)\) time as follows.
Order the time stamped copies of \(\widetilde G\) by
increasing timestamp, breaking ties arbitrarily. 
Since \(\widetilde G\) is a DAG, this order is topological. Therefore, Dijkstra's algorithm can be replaced by a standard linear-time shortest path algorithm for weighted DAGs~\cite{Cormen_book_algorithms}. 
Using the same idea with the $3$-valued flag as above, the query time becomes~\(\mathcal{O}(M)\).

\section{Open questions and future work}\label{sec:conclusion}

Our work opens several directions for further research, both on the algorithmic
side and on the design of data structures for temporal beer paths.  We briefly
discuss some of them below.

     \textbf{FBP and SBP under adjacency list representation.} In this work we showed that it is possible to achieve a speedup for EABP and LDBP if an adjacency list is given as input (compared to an edge stream), assuming non-existence of dominated edges. A natural open question is whether the same is possible for FBP and SBP. This seems challenging, as to the best of our knowledge, there is no faster alternative to Algorithms~\ref{alg:non_dominated_paths} and~\ref{alg:distance_non_dominated_paths} (which we use as subroutines for FBP and SBP respectively) for computing non-dominated paths.
     
   \textbf{One-to-all SBP.}
  In this work we have shown that SBP admits a \emph{one-to-one} algorithm running in time (roughly) $\bigO(M\log d_{max})$, with an edge stream as input. 
  We leave as an open question whether a one-to-all SBP algorithm can be achieved with comparable running time without preprocessing.
  
   \textbf{Multiple types of points of interest.}
  In many applications there may be several classes of different points of interest
  (e.g., gas stations, charging points and rest areas), possibly with different
  constraints or perhaps even priorities.  Extending our framework to handle such cases  appears to be non-trivial, particularly for FBP and SBP.

   \textbf{Interval temporal graphs.}
  An appealing generalization is that of interval temporal
  graphs, where edge availabilities (or beer vertex active times) are described by continuous time intervals instead of discrete instants. Adapting our algorithms to interval representations similar to~\cite{Sahni_interval} seems possible with minor modifications (such as binary searches inside each interval).

  \textbf{Data structures for special-case sublinear beer queries.}
  In the static setting, several works propose data structures that achieve sublinear time queries for shortest beer path queries on restricted graph classes, such as outerplanar graphs (cf.~\cite{Bacic_beer_path}).  It is natural to ask for analogous temporal graph classes (e.g. sparse temporal graphs) for which such a data structure would be possible.










\section{Acknowledgements}
Andrea D'Ascenzo and Giuseppe Italiano have been funded by the Italian Ministry of University and Reseach under PRIN Project n. 2022TS4Y3N - EXPAND: scalable algorithms for EXPloratory Analyses of heterogeneous and dynamic Networked Data. Sotiris Kanellopoulos, Aris Pagourtzis and Christos Pergaminelis have been supported by project MIS 5154714 of the National Recovery and Resilience Plan Greece 2.0 funded by the European Union under the NextGenerationEU Program. Anna Mpanti has been funded by the European Union - Next Generation EU, Mission~4 Component 2 CUP: I83C22000990001.

No artificial intelligence was used at any point during the production of this paper. All of the presented results were conceived exclusively by the authors and the entirety of the paper was written exclusively by the authors.

\printcredits

\bibliographystyle{elsarticle-num}
\bibliography{bibliography}










\end{document}